\documentclass[11pt]{amsart}

\usepackage[hmargin=0.8in,twosideshift=0in,height=8.6in]{geometry}
\usepackage{amssymb,amsthm}
\usepackage{delarray}
\usepackage{natbib}

\usepackage{ifpdf}
\ifpdf
\usepackage[pdftex]{graphicx}
\else
\usepackage[dvips]{graphicx}
\fi

\usepackage{ifpdf}
\usepackage{color}
\definecolor{webgreen}{rgb}{0,.5,0}
\definecolor{webbrown}{rgb}{.8,0,0}
\definecolor{emphcolor}{rgb}{0.95,0.95,0.95}

\usepackage{hyperref}
\hypersetup{%
          colorlinks=true,
          linkcolor=webbrown,
          filecolor=webbrown,
          citecolor=webgreen,
          breaklinks=true}
\ifpdf
\hypersetup{pdftex,
            pdfstartview=FitH, 
            bookmarksopen=true,
            bookmarksnumbered=true
}
\else
\hypersetup{dvips}
\fi

\numberwithin{equation}{section} 

\newtheorem{prop}{Proposition}

\newtheorem{example}{Example}

\newtheorem{definition}{Definition}

{\bf}{\it}
\newtheorem{remark}{Remark}
{\bf}{\it}

\newcommand{\R}{\mathbb{R}}

\title[On the Stickiness Property]{On the Stickiness Property \thanks{We would like to thank
Jean Pierre Fouque, Balgobin Nandram 
Jesus Rodriguez and L. C. G Rogers  for helpful comments.}}

\author{Erhan Bayraktar}
\address[Erhan Bayraktar]{Department of
  Mathematics, University of Michigan, Ann Arbor, MI 48109}
\email{erhan@umich.edu}
\thanks{E. Bayraktar is supported in part by the National Science Foundation under grant number DMS-090625.}

\author{  Hasanjan Sayit}
\address[Hasanjan Sayit] {Department of Mathematics, Worcester Polytechnic Institute.}
\email{hs7@WPI.EDU}



\begin{document}
\begin{abstract} In \cite{Gua} the
notion of stickiness for stochastic processes was introduced. It was also shown that stickiness 
implies absence of arbitrage in a market with proportional transaction costs. In this paper, we
investigate the notion of stickiness further. In particular, we give examples of processes that are not semimartingales but are sticky.
\end{abstract}
\maketitle

\section{Introduction}

In \cite{Gua}, no arbitrage conditions for markets with proportional 
transaction costs were given. In the framework of \cite{Gua} an investor can invest in an asset
$X$ which is a c\'adl\'ag (a French acronym which means right continuous with left limits), adapted, and quasi-left continuous  process. The investor follows an admissible strategy $\theta$, which is an adapted left continuous process of finite variation and satisfies
 \[
 V_t(\theta)= \int_0^{t} \theta_s dX_s -k \int_0^{t}X_s d(D \theta)_s-kX_t |\theta_t| \geq -M
 \]
almost surely, for some $M>0$ and all $t>0$. Here, $D \theta$ denotes the weak derivative of $\theta$ and $|D \theta|$ is the total variation of $D \theta$. Here $V_t$ is the liquidation value of investors portfolio and $k$ is the proportional transaction cost. We say that an admissible strategy is an arbitrage on $[0,T]$ if $V_T(\theta) \geq 0$ and $P(V_T>0)>0$.

If $X_t$ is \emph{sticky},  \cite{Gua} showed that there is no arbitrage with respect to $e^{X_t}$ when there are transaction costs. 
Let us recall the definition of stickiness from \cite{Gua}. 
\begin{definition} \label{defi1} A progressively measurable process $X$ is
  sticky with respect to the filtration $\mathbb{F}$ and the probability measure $P$ if for all 
$\epsilon>0, T>0$ and all stopping times $\tau$ such that $P(\tau<
T)>0$, we have that:
\[
P\left(\sup_{t\in [\tau, T]}|X_{\tau}-X_t|<\epsilon, \tau <T\right)>0.
\]
\end{definition}
Stickiness is satisfied by any continuous 
process that has full support
in the Wiener space (Proposition 4.1
of \cite{Gua}),  and also by 
a large class of Markov processes (Proposition 3.1 of \cite{Gua}), including 
diffusions and L\'{e}vy processes. Then, by exploiting the characterization of the support of Gaussian processes, it was shown that fractional Brownian motion (fBm) has full support in the 
Wiener space which implies its stickiness. As a result when there are proportional transaction costs, even when the asset prices are modeled by geometric fBm, there is no arbitrage opportunity. (FBm modulated markets allow for
arbitrage if there are no transaction costs, see e.g. \cite{Bayraktar} and \cite{LCGR}.) We should point out that introducing other types of frictions into the market, in particular by putting suitable restrictions on allowable trading strategies, we can have price processes that are not semimartingales but still have no-arbitrage in the market (see e.g. \cite{BS09} and \cite{HPJ}).

The purpose of our note is to study the notion of stickiness further and give other examples of processes that are sticky. We will first show that continuous functions of stick processes remain sticky (see Proposition~\ref{stickymonotone}). Next we show that bounded time changes of sticky processes remain sticky (see Proposition~\ref{timechange}). Using this result, we then prove a stickiness result for semimartingales (see Proposition~\ref{time}). Thanks to these results we are able to give examples of sticky processes that are not semimartingales. Therefore, our examples allow for arbitrage in the frictionless markets, but do not allow for arbitrage in the markets with transaction costs. Other papers, which give examples of processes satisfying no arbitrage conditions in the markets with transaction costs include \cite{MR2398764} and \cite{MR2432181}.

Throughout the paper we assume we are given a complete, filtered
probability space $(\Omega,\mathcal{F},P,\mathbb{F}=(\mathcal{F})_{t \geq 0})$ satisfying the 
``usual hypotheses" (i.e., the filtration 
$\mathbb{F}$ is right continuous, and $\mathcal{F}_0$ contains all of
the $P$ null sets of $\mathcal{F}$).

\section{Main Results}

Before stating our main results, we give alternative formulations of stickiness. These formulations will be useful in the proofs of our main results.

\begin{prop}\label{prop:new-char} Let $X$ be a progressively measurable process with respect to the filtration $\mathbb{F}$. Then the following statements are equivalent
\begin{enumerate}
\item[(a)] $X$ is sticky.
\item[(b)] For any bounded stopping time $\tau$ of $\mathbb{F}$ and any $A\in \mathcal{F}_{\tau}$ with $P(A)>0$, we have $P(A\cap \{\sup_{t\in[\tau, T] }|X_{\tau}-X_t|<\epsilon\})>0$ for any $\epsilon>0$ and any 
$T\geq \tau$ a.s.
\item[(c)] For any bounded stopping time $\tau_0$ of $\mathbb{F}$ and any number $\delta>0$, the stopping time
$\tau_1=\inf\{t\geq \tau_0: |X_t-X_{\tau_0}|>\delta\}$ is unbounded on $A$ for any $A\in \mathcal{F}_{\tau_0}$ with $P(A)>0$.
\end{enumerate}
\end{prop}

\begin{proof} $(a)\Rightarrow (b)$: Take any
bounded stopping time $\tau$ and any $T\geq \tau$ a.s. Let
$A\in \mathcal{F}_{\tau}$ with $P(A)>0$. Define $\tau^A=\tau$ on $A$ and
$\tau^A=T$ on $A^c$. Then we have 
$\{\sup_{t \in [\tau,
  T]}|X_{\tau}-X_t|<\epsilon, \tau^A<T\}\subset A\cap \{\sup_{t\in [\tau,
  T]}|X_{\tau}-X_t|<\epsilon\}
$. Now if $P(\tau^A<T)>0$, then since $X$ is sticky we have $P(\{\sup_{t \in [\tau,
  T]}|X_{\tau}-X_t|<\epsilon, \tau^A<T\})>0$. If $P(\tau^A<T)=0$, then $\tau=T$ a.s., so 
$A\cap \{\sup_{t\in [\tau,
  T]}|X_{\tau}-X_t|<\epsilon\}=A$. Therefore in both cases we have
$P(A\cap \{\sup_{t\in [\tau, T]}|X_{\tau}-X_t|<\epsilon\})>0$.

$(b)\Rightarrow (c)$: Let us assume that there exists a bounded stopping time $\tau_0$ and a real number $\delta>0$, such that
$\tau_1=\inf\{t\geq \tau_0: |X_t-X_{\tau_0}|>\delta\}$ is bounded on some $A\in \mathcal{F}_{\tau_0}$ with $P(A)>0$. Then since $X$ is c\'adl\'ag we have that $|X_{\tau_1}-X_{\tau_0}| \geq \delta$ on $A$.
As a result, for any real number $T$ with $T>\tau_1$ on $A$, we have  $P(A\cap\{\sup_{t\in [\tau_0, T]}|X_t-X_{\tau_0}|<\frac{\delta}{2}\})=0$. 

$(c)\Rightarrow (a)$: Let $\tau$ be any stopping time and $T$ be any real number with $P(\tau<T)>0$. Let $A=\{\tau<T\}$, then $A\in \mathcal{F}_{\tau}$ and $P(A)>0$. Let $\tau^A=\tau 1_{A}+ T 1_{A^{c}}$, which is a bounded stopping time. Now, for any $\epsilon>0$, since $\tau_1=\inf\{t\geq \tau^A: |X_t-X_{\tau^A}|>\frac{\epsilon}{2}\}$ is unbounded on $A$, we have that there exists $A_1 \subset A$ with $P(A_1)>0$ such that
 $\sup_{t\in[\tau^A, T]}|X_t-X_{\tau^A}|<\epsilon$ on $A_1$; from which the result follows.
\end{proof}

Next, we show that stickiness is preserved under composition with continuous functions.

\begin{prop}\label{stickymonotone} Let $X$ be a progressively measurable, c\`adl\`ag process
that takes values on $(a, b)$, for $a,b \in \bar{\R}=[-\infty,\infty]$.
Let $f$ be a continuous function on
 $(a, b)$. If $X$ is sticky, then $f(X)$ is also sticky.
\end{prop}

\begin{proof} Let us first prove the statement when $a,b$ are bounded.
Due to the equivalence of (a) and (b) in Proposition~\ref{prop:new-char}, we need to show for any bounded stopping time $\tau$ and
  any $A\in \mathcal{F}_{\tau}$ with $P(A)>0$, we have 
that $P(A\cap \{\sup_{t\in [\tau, T]}|f(X_t)-f(X_{\tau})|< \epsilon\})>0$
for any $\epsilon >0$ and $T$ with $\tau\le T$ a.s. Since $X_{\tau}$
takes values in $(a, b)$ and $P(A)>0$, for sufficiently large 
$n_0\in \mathbb{N}_+$, the
event $B=A\cap \{X_{\tau}\in [a+\frac{1}{n_0}, b-\frac{1}{n_0}]\}$ has
positive probability and $B \in \mathcal{F}_{\tau}$. The function $f$ is continuous
on $(a,b)$, so it is uniformly continuous on
$[a+\frac{1}{n_0}, b-\frac{1}{n_0}]$. This means that for a given $\epsilon$,  there exists $\delta_0>0$ such that whenever  $|x-y|<\delta_0$ and $x,y \in [a+\frac{1}{n_0}, b-\frac{1}{n_0}]$ we have $|f(x)-f(y)|<\epsilon$.  Let $\delta=\min\{\delta_0,1/n_0\}$. Since $X$ has the sticky property, we have 
$P(B\cap \{\sup_{t\in [\tau, T]}|X_t-X_{\tau}|<\delta \})>0$. Now as a result of the uniform continuity of $f$ we have that $B\cap \{\sup_{t\in [\tau, T]}|X_t-X_{\tau}|<\delta \}\subset 
B\cap \{\sup_{t\in [\tau, T]}|f(X_t)-f(X_{\tau})|<\epsilon\}$, from which 
the result follows since $B \subset A$. When $a=-\infty$ and/or $b=\infty$, the above proof
can be adjusted by replacing $a+1/n_0$ with $a_n\downarrow -\infty$ and/or 
$b-1/n_0$ with $b_n\uparrow \infty$.   
\hfill  \end{proof}

\begin{example} In \cite{Gua}, the stickiness of $|B_t|$, where $B_t$ is Brownian motion, was explained by the strong Markov property of the $|B_t|$. Using Proposition~\ref{stickymonotone} the stickiness of $|B_t|$ can be explained by the continuity of $x \to |x|$ and the stickiness of $B_t$.

\end{example}

\begin{example} Let $B_t^H$ be the fractional Brownian motion with Hurst parameter $H\in [0, 1]$. In \cite{Gua}, it was shown that $B_t^H$ is sticky with respect to its natural filtration. 
Proposition \ref{stickymonotone} can be used to conclude that the process $f(B_t^H)$ is sticky for any continuous function $f$. 
\end{example}

\subsection{Time changed Sticky Processes}

Our next result shows that the stickiness is preserved under bounded time changes. We should mention that the boundedness of the stopping times is crucial as the next example shows.

\begin{example}
Let $\{B_t\}$ be the standard one dimensional Brownian motion. For $s>0$, let $T_s=\inf\{t\geq 0: B_t=s\}$ be the passage time of $B$ to level $s$. It is well-known that $\{T_s\}$ are unbounded stopping times. And it is clear that the time changed process $B_{T_s}=s$ is not sticky.
\end{example}

However, for bounded time changes we can state the following result:

\begin{prop} \label{timechange} Let  $X$ be a continuous process that
 is progressively measurable and sticky with respect to the filtration $\mathbb{F}$. Let $
(\nu_{t})_{t\geq 0}$ be a family of $\mathbb{F}$ stopping times such that $t \rightarrow \nu_{t}$ is continuous, non-decreasing almost surely. We also assume that $\nu_t$ is bounded almost surely for each $t \geq 0$ and $\nu_0=0$.
Let $\widetilde{X}_t=X_{\nu_{t}}$, and $\widetilde{{
\mathcal{F}}}_{t}={\mathcal{F}}_{\nu_{t}}$, for $t\geq 0$. Then 
$\widetilde{X}$ is sticky with respect to the filtration $\widetilde
{\mathbb{F}}=(\widetilde{\mathcal{F}}_t)_{t \geq 0}$. 
\end{prop}

\begin{proof} Due to the equivalence of (a) and (c) in Proposition~\ref{prop:new-char}, we need to show for any
  bounded stopping time $\tau$ of $\widetilde{\mathbb{F}}$ and any $\delta
  >0$, the $\widetilde{\mathbb{F}}$-stopping time 
$\theta=\inf \{t\geq \tau:|\widetilde {X}_t-\widetilde {X}_{\tau}|>\delta \}$
is unbounded on
each $A\in \widetilde {\mathcal{F}}_{\tau}$ with $P(A)>0$. Assume the contrary, i.e., $\theta$ is bounded on some  $A\in \widetilde {\mathcal{F}}_{\tau}$ with
$P(A)>0$ and for some $\delta >0$. 
Let 
$\theta_0=\inf\{t\geq \tau:
|\widetilde{X}_t-\widetilde{X}_{\tau}|>\frac{\delta}{2}\}$. Then, the $\widetilde{\mathbb{F}}$-stopping time $\theta_0$ is
bounded on $A$ since $\theta$ is bounded on $A$. 
Let $\theta_0^A=\theta_0$, $\tau^A=\tau$  on $A$
and $\theta_0^A=\infty$, $\tau^A=\infty$  on $A^c$.
Choose a deterministic number $k$
such that the event $A_1=\{\tau^A <k< \theta_0^A \}$ has positive 
probability. Such a number exists because $\tau^A <\theta_0^A$ on
$A$. Also we have $A_1\subset A$ and $A_1\in \widetilde{\mathcal{F}}_k$. Since 
$k<\theta_0$ on $A_1$, we have 
$|\widetilde X_{\tau}-\widetilde X_k|<\frac{\delta}{2}$ on $A_1$. 
Since $|\widetilde X_{\theta}-\widetilde X_{\tau}|\geq \delta$ on $A_1$ and 
$|\widetilde X_{\theta}-\widetilde X_{\tau}|\le |\widetilde X_{\theta}-\widetilde
X_k|+|\widetilde X_k- \widetilde X_{\tau}|$ we have that $|\widetilde
X_{\theta}-\widetilde X_k|> \frac{\delta}{2}$ on $A_1$. Since $\theta$ is bounded, $\nu_{\theta}$ is bounded almost surely, and so the last fact implies that the
stopping time 
$\theta_1=\inf\{t\geq \nu_k: |X_t-X_{\nu_k}|> \frac{\delta}{2}\}$ is
bounded on $A_1-N \in \mathcal{F}_{\nu_k}$ (in which $N$ is a null set), which contradicts the 
stickiness of $X$ (by Proposition~\ref{prop:new-char}). 
\hfill  \end{proof}

\begin{example}
Time changed fBm $B^H_{\nu_t}$, for $\nu_t$ satisfying the assumptions of Proposition~\ref{timechange}, admits arbitrage opportunities in frictionless markets (see e.g. \cite{Bayraktar}). The above proposition tells us that the same model does not admit arbitrage opportunities when there are transaction costs in the market.
\end{example}
 
\subsection{Further Examples of Sticky Processes}

Next we use Proposition~\ref{timechange} to state a stickiness result for semimartingales. This result will then be used to construct sticky processes, which are not semimartingales. The following remark will be useful in the proof of this result.

\begin{remark}\label{stickyremark} It directly follows from Definition~\ref{defi1} that if a progressively measurable process $X$ is sticky with respect to a filtration $\mathbb{F}$ and 
 $\mathbb{G}$ is a subfiltration of $\mathbb{F}$ and that $X$ is progressively measurable with respect to
$\mathbb{G}$, then $X$ is sticky with respect to the filtration $\mathbb{G}$. Also,  
if $Q$ is an equivalent measure to the 
original measure $P$, then the stickiness of $X$ under $P$ is equivalent to the
stickiness of $X$ under $Q$. 
\end{remark}

\begin{prop}\label{time} Let $X$ be a continuous $(\mathbb{F},P)$-semimartingale 
that admits an equivalent local martingale measure $Q$. Assume that $X_0=0$. 
Let $f$ be a determinstic continuous function on $[0, \infty)$ with $f(0)=0$.
If $[X, X]_t$ is bounded for each $t>0$ and $\lim_{t \rightarrow \infty}  [X,X]_t=\infty$, then
$X_t-f([X,X]_t)$ is sticky with
respect to the filtration $\mathbb{F}$.
 \end{prop}

\begin{proof} Due to Remark \ref{stickyremark} we only need to show
  $X_t-f([X,X]_t)$ is sticky with respect to
  $(\mathbb{F}, Q)$. Let $T_t=\inf\{s\geq 0: [X,X]_s >t \}$. Then since $X$ is a continuous local martingale under $Q$, by Theorem 1.6 in Chapter V of \cite{RY}, $B_t=X_{T_t}$ is an 
$\widetilde{\mathcal{F}}_t=\mathcal{F}_{T_t}$-Brownian motion under $Q$ (The assumption $\lim_{t \rightarrow \infty}  [X,X]_t=\infty$ was needed to apply this result). 
Since $f$ is a continuous deterministic function null at zero, the
process $B_t+f(t)$ is sticky with respect to $\left(\widetilde{\mathcal{F}}_t\right)_{t \geq 0}$ and $Q$ (also see example 4.2 of
\cite{Gua}). We have the obvious relation 
$X_t+f([X,X]_t)=B_{[X,X]_t}+f([X,X]_t)$. Since $[X,X]_t$ is  stopping time
for the filtration $(\widetilde{\mathcal{F}}_t)$, $X_t+f([X,X]_t)$ is a time
change of $B_t+f(t)$. So by Proposition \ref{timechange},
$X_t-f([X,X])$ is sticky with respect to the filtration
$(\widetilde{\mathcal{F}}_{[X,X]_t})_{t \geq 0}$. Now, we will show that $\mathbb{F}$ is 
sub-filtration of $(\widetilde{\mathcal{F}}_{[X,X]_t})$, i.e. for any
$A\in \mathcal{F}_t$, we want to show that $A\in \widetilde{\mathcal{F}}_{[X, X]_t}$. This
is equivalent to showing $A\cap \{[X, X]_t\le s\}\in 
\widetilde{\mathcal{F}}_s=\mathcal{F}_{T_s}$ for all $s\geq 0$. And this is
equivalent to showing $A\cap \{[X, X]_t\le s\}\cap \{T_s \le u \}\in 
\mathcal{F}_{u}$, for all $s\geq 0, u \geq 0$. But since 
$\{[X, X]_t\le s\}=\{T_s\geq t\}$, all we need to show is
$A\cap \{T_s\geq t\}\cap \{T_s\le u\}\in \mathcal{F}_{u}$. 
Since $A\cap \{T_s\geq t\}\in \mathcal{F}_{T_s}$, the result follows
by the definition of $\sigma$-algebra $\mathcal{F}_{T_s}$. 
Finally, since $X_t-f([X,X]_t)$ is sticky with respect to the filtration  
$(\widetilde{\mathcal{F}}_{[X,X]_t})_{t \geq 0}$ and $\mathbb{F}$ is subfiltration of
$(\widetilde{\mathcal{F}}_{[X,X]_t})_{t \geq 0}$, the result 
follows by Remark \ref{stickyremark}.
 \end{proof}

The next example demonstrates that not all semimartingales are sticky. 

\begin{example} Let $X_t=\int_0^tH_sdB_s$, where $B$ is a standard Brownian motion, and let 
\[
H_s=\left \{
\begin{array}{ll}
\frac{1}{1-s} & \mbox{when $s \le \tau^{\star}$}, \\ 
1 & \mbox{when $s> \tau^{\star}$},
\end{array}
\right.
\]
in which $\tau^{\star}=\inf\{t \geq 0: |\int_0^t\frac{1}{1-s}dB_s|>2\}$. Observe that $\tau^{\star}\le 1$. 
Choosing $\tau=0, A=\Omega, T=1, \varepsilon=1$
we have that $P(A\cap\{\sup_{t\in [\tau, T]}|X_t-X_{\tau}|\geq \varepsilon \} )=0$. Due to Proposition ~\ref{prop:new-char} we conclude that the martingale $X$ is not sticky.
\end{example}

The main use of Proposition~\ref{time} is to create examples of sticky processes that are not semimartingales. For example it can be used along with Proposition~\ref{stickymonotone} to give the following example. 
\begin{example}
Let $M$ be a continuous local martingale with $[M, M]_t$ is bounded for each $t$, $M_0=0$ and $\lim_{t \to \infty}[M,M]_t=\infty$. The process $X_t=|M_t|^{\frac{1}{3}}$ is not a semimartingale; see
Theorem 72 on page 221 of \cite{Pro}. However, it follows from Propositions~\ref{stickymonotone} and \ref{time} that $X$ is sticky.
\end{example}

It should be noted that $X_t-f([X,X]_t)$ in Proposition~\ref{time} does not have to be a semi-martingale as we see in the next example.
\begin{example}
Let $B_t$ be a Brownian motion and define $\tau=\inf\{t\geq 0: |B_t|\geq 1\}$.
Consider $X_t=\int_0^tB_{s\wedge \tau}dB_s-h(\int_0^tB_{s\wedge \tau}^2ds)$ in which 
 $h(x)=xcos\frac{\pi}{x}$ when $x\neq 0$ and $h(0)=0$.
The process $X$ is not semimartingale (since the total variation of $h$ is unbounded on any interval $[0,b]$ for $b>0$); however, it is sticky thanks to Proposition~\ref{time}.
\end{example} 

\section*{Acknowledgment} We would like to thank the three anonymous referees for their feedback, which helped us improve our paper in significant ways.

{\small
 \bibliographystyle{plain}
\bibliography{references}
}

\end{document}